\documentclass[journal,12pt,onecolumn]{IEEEtran}

\usepackage{setspace}
\doublespace
\usepackage{graphics}
\usepackage{rotating}
\usepackage{amsfonts}
\usepackage{amsmath}
\usepackage{subfigure}

\graphicspath{{Figures/}} \DeclareGraphicsExtensions{.eps,.ps}

\begin{document}


\title{Bit-Interleaved Coded Multiple Beamforming \\with Imperfect CSIT}


\author{Ersin Sengul~\IEEEmembership{Member, IEEE},
    Hong Ju Park, \\and
        Ender Ayanoglu~\IEEEmembership{Fellow, IEEE}\\
\authorblockA{Center for Pervasive Communications and Computing\\
Department of Electrical Engineering and Computer Science\\
University of California, Irvine\\
Irvine, California 92697-2625\\
Email: esengul@uci.edu, hjpark@uci.edu, ayanoglu@uci.edu}}


\maketitle

%


\begin{abstract}
This paper addresses the performance of bit-interleaved coded
multiple beamforming (BICMB) \cite{akayWCNC06}, \cite{akayTC06BICMB}
with imperfect knowledge of beamforming vectors. Most studies for
limited-rate channel state information at the transmitter (CSIT)
assume that the precoding matrix has an invariance property under an
arbitrary unitary transform. In BICMB, this property does not hold.
On the other hand, the optimum precoder and detector for BICMB are
invariant under a {\em diagonal\/} unitary transform. In order to
design a limited-rate CSIT system for BICMB, we propose a new
distortion measure optimum under this invariance. Based on this new
distortion measure, we introduce a new set of centroids and employ
the generalized Lloyd algorithm for codebook design. We provide
simulation results demonstrating the performance improvement
achieved with the proposed distortion measure and the
codebook design for various receivers with linear detectors. We show that although these receivers have the same performance for perfect CSIT, their performance varies under imperfect CSIT.\\\\
\end{abstract}

\section{Introduction}\label{sec:intro}
It is well-known that multiple-input multiple-output (MIMO) systems
enhance the throughput of wireless systems, with an increase in
reliability and spectral efficiency \cite{foschini96},
\cite{foschini98}, \cite{raleigh98}. While the advantages of MIMO
architectures are attainable when only the receiver side knows the
channel, the potential gains can be further improved when the
transmitter has some knowledge of the channel, which is known as
channel state information at the transmitter (CSIT). CSIT can be
used to improve diversity order or array gain of a MIMO wireless
system. In this work, we are interested in multi-stream precoding to
achieve MIMO spatial multiplexing. In this paper ``spatial
multiplexing order'' refers to the number of multiple symbols
transmitted, as in \cite{paulraj}. This term is different than
``spatial multiplexing gain'' defined in \cite{tse03}. Throughout
the paper, we will employ the terminology \emph{single beamforming}
and \emph{multiple beamforming} to refer to single- and multi-stream
precoding, respectively \cite{palomarThesis},
\cite{akay05TCOM_BeamformingLetter}.

Precoders based on perfect CSIT are designed in
\cite{sampath01tcom}, \cite{scaglione02designs}, \cite{palomar03}
for many different design criteria. The majority of the designs
include the channel eigenvectors which are obtained through the
singular value decomposition (SVD) of the channel. It is
well-known that it may not be practical to have perfect CSIT. In
this paper, we will design a system with limited CSIT when the
channel obeys the standard block fading (quasi-static) model. In
this model, the channel may change from block to block, but
remains constant during the transmission of a block. This model is
commonly used in the design and simulation of broadband wireless
systems.

Recently, limited CSIT feedback techniques have been introduced to
achieve a performance close to the perfect CSIT case. In these, a
codebook of precoding matrices is known both at the transmitter and
receiver. The receiver selects the precoding matrix that satisfies a
desired criterion, and only the index of the precoding matrix is
sent back to the transmitter. Initial work on limited feedback
systems concentrated on single beamforming where a single symbol is
transmitted along a quantized version of the optimal beamforming
direction. Authors of \cite{narula98} analyzed single beamforming in
a multi-input single-output (MISO) setting where they designed
codebooks via the generalized Lloyd algorithm. The relationship
between codebook design for quantized single beamforming and
Grassmannian line packing was observed in \cite{mukkavilli01LF},
\cite{love_grasmannian03} for i.i.d. Rayleigh fading channels. This
connection enabled the authors in \cite{mukkavilli01LF},
\cite{love_grasmannian03} to leverage the work already carried out
for optimal line packing in the mathematics literature. Authors in
\cite{xiaWelch} proposed a systematic way of designing good
codebooks for single beamforming inspired from \cite{hochwald99}.
Rate-distortion theory tools were used in \cite{xia_LF_SB_VQ} to
analyze single beamforming performance when the generalized Lloyd
algorithm is used. Random vector quantization (RVQ) technique, where
a random codebook is generated for each channel realization, was
used to analyze single beamforming in an asymptotic scenario
\cite{santipachRVQ}. Later, results were generalized to multiple
beamforming \cite{love03LF-SM}, \cite{rohVQ_MB}. The results in
\cite{love03LF-SM} showed that there is a relation between codebook
design for multiple beamforming and Grassmannian subspace packing.
However the results in \cite{love03LF-SM} are specific to uncoded
multiple beamforming. Most papers considered the unitary or
semi-unitary constraint on the precoder since the optimal linear
precoder is unitary with perfect CSIT for linear receiver
architectures \cite{palomarThesis}. In such a case, it is possible
to exploit the properties of unitary matrices and parameterize the
optimal precoder into a set of angles to be quantized
\cite{sadrabadiGivens}, \cite{roh04Givens}. hat with the

It has been shown that for a MIMO system with $N$ transmit and $M$
receive antennas, it is possible to achieve full spatial diversity
of $NM$, independent of the number of streams $1\leq S
\leq\min(N,M)$ transmitted over quasi-static Rayleigh flat fading
channels. One possible system achieving this limit is the so-called
bit-interleaved coded multiple beamforming (BICMB)
\cite{akayWCNC06}, \cite{akayTC06BICMB}. Design criteria for the
interleaver and the convolutional encoder which guarantee full
diversity and full spatial multiplexing are provided in
\cite{akayWCNC06}, \cite{akayTC06BICMB}. Previously, bit-interleaved
coded modulation (BICM) \cite{bicm}, \cite{zehavi} was employed in
single- and multi-antenna systems without utilizing CSIT
\cite{akay04VTC-F-BO}, \cite{akay04VTC-F-BSO}, \cite{lee03},
\cite{rende&wong}. In general, BICMB requires perfect knowledge of
only channel eigenvectors at the transmitter, i.e., does not need
the channel gains (eigenvalues) at the transmitter. It has linear
detection complexity and needs a simple soft-input Viterbi decoder.
It also achieves full diversity without any adaptation for the
number of streams.

In this paper, the goal is to design a limited feedback scheme for
BICMB. We first deal with codeword selection criterion assuming that
there is already a given codebook. We provide a new optimal
distortion measure for the selection of the best precoder from the
codebook. This new distortion measure is due to the non-uniqueness
property of the SVD \cite{horn}. We then calculate a centroid for
this new distortion measure. We analyze the performance of the
proposed distortion measure for different receiver structures
through extensive simulations. For comparison purposes, we first use
a randomly generated codebook. Next, we utilize the generalized
Lloyd's algorithm \cite{lindeLGB} to design better codebooks. For
this new codebook, we employ the minimum mean square error (MMSE)
and the zero-forcing (ZF) receivers as well as a new receiver.

\textbf{Notation:} $N$ is the number of transmit antennas, $M$ is
the number of receive antennas. The symbol $S$ denotes the total
number of symbols transmitted at a time (spatial multiplexing order,
in other words the total number of streams used). The superscripts
$(\cdot)^\dag$, $(\cdot)^H$, $(\cdot)^T$, $(\cdot)^*$, and the
symbol $\forall$ denote the pseudoinverse, Hermitian, transpose,
complex conjugate, and for-all respectively.\\

\section{System Model}\label{sec:QBICMBsystem}

In the limited feedback context, authors of \cite{love03LF-SM}
showed that, in their uncoded system, for both the ZF and the MMSE
receiver the optimal precoder is in the form of $\mathbf{VQ}$,
where is $\mathbf{V}$ is the channel right singular matrix and
$\mathbf{Q}$ is any unitary matrix. This characterization enabled
authors to see the direct relation between codebook design for
multiple beamforming and Grassmannian subspace packing. However,
as we will show, in our system, multiplication of the channel
right singular matrix $\mathbf{V}$ with a general unitary matrix
$\mathbf{Q}$, and employing $\mathbf{VQ}$ as the precoding matrix
causes performance degradation. A new selection criterion and
codebook design procedure is needed for limited feedback in BICMB.

In BICMB, the output bits of a binary convolutional encoder are
interleaved and then mapped over a signal set $\chi \subseteq
\mathbb{C}$ of size $|\chi|=2^m$ with a binary labeling map $\mu :
\{0,1\}^m \to \chi$. We use the same interleaver that was previously
employed for the perfect CSIT case in \cite{akayTC06BICMB}. The
interleaver is not unique and not necessarily the optimal one, but
satisfies the design criterion and enables the system to have full
diversity when perfect CSIT is available. Gray encoding is used to
map the bits onto symbols. During transmission, the code sequence
$\underbar{c}$ is interleaved by $\pi$, and then mapped onto the
signal sequence $\underbar{x} \in \chi$.

Let $\mathbf{H}$ denote the quasi-static, flat fading $M \times N$
MIMO channel, where $M$ and $N$ are the number of receive and
transmit antennas, respectively, and assume perfect timing,
synchronization and sampling. In this paper, we assume that the
transmitter employs multiple beamforming prior to the transmission
of the complex baseband symbols. When $S$ symbols are transmitted
at the same time, the system input-output relation between
transmitted and received baseband complex symbols can be written
as
\begin{align}
\mathbf{y} &= \mathbf{HV}_L\mathbf{x}+\mathbf{n}
\label{eq:InOutBeamLF}
\end{align}
where $\mathbf{x}$ is an $S \times 1$ vector of symbols to be
transmitted, $\mathbf{n}$ is an $M \times 1$ additive white
Gaussian noise vector whose elements have zero mean and variance
$N_0 = N/SNR$, and $\mathbf{V}_L$ is an $N\times S$ precoding
matrix, which is dependent on the instantaneous channel
realization. The total power transmitted is scaled as $N$. The
channel matrix elements are modeled as i.i.d. zero-mean,
unit-variance complex Gaussian random variables. Consequently, the
received average signal-to-noise ratio is $SNR$.

We assume that the receiver selects a precoder matrix from a finite
set of beamforming matrices and sends the index of the selected
precoder through an error-free feedback link without any delay.
Precoded symbols are transmitted over the channel and at the
receiver a linear equalizer is used as a detector prior to the
Viterbi decoder. Our aim is to investigate the effects of imperfect
CSIT on the BICMB system compared to the perfect CSIT scenario and
therefore, we concentrate on a linear detector followed by soft
input non-iterative Viterbi decoder as in \cite{akayWCNC06},
\cite{akayTC06BICMB}. In this paper we do not consider nonlinear
detectors or iterative decoding techniques.

The bit interleaver of BICMB can be modeled as $\pi : k' \to
(k,s,i)$ where $k'$ denotes the original ordering of the coded bits
$c_{k'}$, $k$ denotes the time ordering of the signals $x_{k,s}$
transmitted, $s$ denotes the subchannel used to transmit $x_{k,s}$,
and $i$ indicates the position of the bit $c_{k'}$ on the symbol
$x_{k,s}$. Let $\chi_b^i$ denote the subset of all signals $x \in
\chi$ whose label has the value $b \in \{0,1\} $ in position $i$.
The bit metrics, i.e., $\gamma^i (y_{k,s},c_{k'})$, are dependent on
the receiver structure and will be revisited in Section
\ref{sec:Detectors_BICMB_LF}. The Viterbi decoder at the receiver
makes decisions according to the rule
\begin{equation}
\underline{\boldsymbol{\hat{c}}} = \arg
\min_{\underline{\boldsymbol{c}} \in \mathcal{C}} \sum_{k'}
\gamma^i(y_{k,s},c_{k'}). \label{eq:MLrule}
\end{equation}

\section{Bit-Interleaved Coded Multiple Beamforming}\label{sec:LF_BICMB}

\subsection{Background on SVD}\label{sec:SVD}

As stated previously, the work in this paper depends on the fact
that SVD has an invariance property under diagonal unitary
transformation. We provide a formal description of this fact below
\cite{horn}.

\newtheorem{SVD}{Theorem}
\begin{SVD}
If $\mathbf{H} \in \mathbb{C}^{M\times N}$ has rank $k$, then it
may be written in the form $\mathbf{H = U\Sigma V}^H$, where
$\mathbf{U}$ and $\mathbf{V}$ are unitary matrices whose columns
are the left and right singular vectors of $\mathbf{H}$. The
matrix $\mathbf{\Sigma}=[\sigma_{ij}] \in \mathbb{R}^{M\times N}$
has $\sigma_{ij}=0$ for all $i\neq j$, and $\sigma_{11}\geq \cdots
\geq \sigma_{kk}> 0$, and $\sigma_{k+1,k+1} = $ $\cdots
=\sigma_{qq}=0$, where $q = \min(N,M)$. The numbers
$\sigma_{ii}\equiv\sigma_i, i=1,2, \ldots,q$ are the nonnegative
square roots of the eigenvalues of $\mathbf{HH}^H$, and hence are
uniquely determined. The columns of $\mathbf{U}$ are eigenvectors
of  $\mathbf{H}\mathbf{H}^H$ and the columns of $\mathbf{V}$ are
eigenvectors of  $\mathbf{H}^H\mathbf{H}$. If $N\leq M$ and if
$\mathbf{H}^H\mathbf{H}$ has distinct eigenvalues, then
$\mathbf{V}$ is determined up to a right diagonal unitary matrix
$\mathbf{D} =$ diag$(e^{j\theta_1}, e^{j\theta_2}, \ldots,
e^{j\theta_N})$ with all $\theta_i \in [0,2\pi)$; that is, if
$\mathbf{H = U}_1\mathbf{\Sigma}
\mathbf{V}_1^H=\mathbf{U}_2\mathbf{\Sigma}\mathbf{V}_2^H$, then
$\mathbf{V}_2 = \mathbf{V}_1 \mathbf{D}$.
\end{SVD}

\begin{proof}
See \cite{horn}.
\end{proof}

The conditions of the theorem above hold for the system in this
paper, and therefore there are infinitely many right singular
matrices for a given channel realization. Note that when $S\leq q$
streams are transmitted, the first $S$ columns of $\mathbf{V}$,
i.e., $\overline{\mathbf{V}}$, are employed. Therefore, if
$\mathbf{H = U}_1\mathbf{\Sigma}
\mathbf{V}_1^H=\mathbf{U}_2\mathbf{\Sigma}\mathbf{V}_2^H$, then
$\overline{\mathbf{V}}_2 = \overline{\mathbf{V}}_1 \mathbf{D}$ and
$\overline{\mathbf{U}}_2 = \overline{\mathbf{U}}_1 \mathbf{D}$,
where $\mathbf{D}$ is any $S\times S$ diagonal unitary matrix.

\subsection{Selection Criteria}\label{sec:SC_BICMB_LF}

In this section, we assume that there exists a codebook and we
wish to find a criterion to choose the best approximation to
$\mathbf{V}$ from the codebook $\mathbb{V}=
\{\mathbf{\hat{V}}_i\}_{i=1}^C$, where $C$ is the codebook size.
One could potentially use the well-known Euclidean metric, however
the property described in Theorem 1 complicates the problem.\\

\emph{Selection Criterion -  Euclidean (SC-E) :} The receiver
selects $\mathbf{V}_L$ such that
\begin{align}
\mathbf{V}_L= \mbox{ }\underset{\mathbf{\hat{V}}_i \in
\mathbb{V}}{\arg\min} \mbox{ }||\mbox{
}\mathbf{V}-\mathbf{\hat{V}}_i\mbox{ }||_F^2.
\label{eq:Eclidean_V}
\end{align}
This selection criterion aims to find the codebook element closest
to the optimal beamforming matrix $\mathbf{V}$. It can be argued
that this criterion asymptotically diagonalizes the system as the
number of feedback bits goes to infinity.

However, the property in Theorem 1 makes straightforward
application of (\ref{eq:Eclidean_V}) nonpractical. This can be
explained with the aid of Figure \ref{fig:BICMB_sets}. Assume that
an application of SVD for a given instantiation of the
$\mathbf{H}$ matrix yields a $\mathbf{V}$ matrix. Assume that when
$\mathbf{V}$ is multiplied by all diagonal unitary matrices
$\mathbf{D}$, one gets the set $\mathbb{S}_{\mathbf{V}}$ in Figure
\ref{fig:BICMB_sets}. It should be clear that the closest member
of $\mathbb{V}$ to $\mathbf{V}$ is not necessarily the closest
member of $\mathbb{V}$ to $\mathbb{S}_{\mathbf{V}}$. As a result,
one needs to modify (\ref{eq:Eclidean_V}) such that the minimum
distance between two sets $\mathbb{V}$ and
$\mathbb{S}_{\mathbf{V}}$ can be calculated. A way to accomplish
this is
\begin{align}
\mathbf{V}_L= \mbox{ }\underset{{\mathbf{\hat{V}}_i \in
\mathbb{V}, \mathbf{D} \in \mathbb{D}}} {\arg\min} \mbox{
}||\mbox{ }\mathbf{VD}-\mathbf{\hat{V}}_i\mbox{ }||_F^2
\label{eq:Eclidean_VD}
\end{align}
where $\mathbb{D}$ stands for the set of all \emph{diagonal
unitary}
matrices.\\
\begin{figure}[pg6]
\centering \includegraphics[width =0.60 \linewidth]{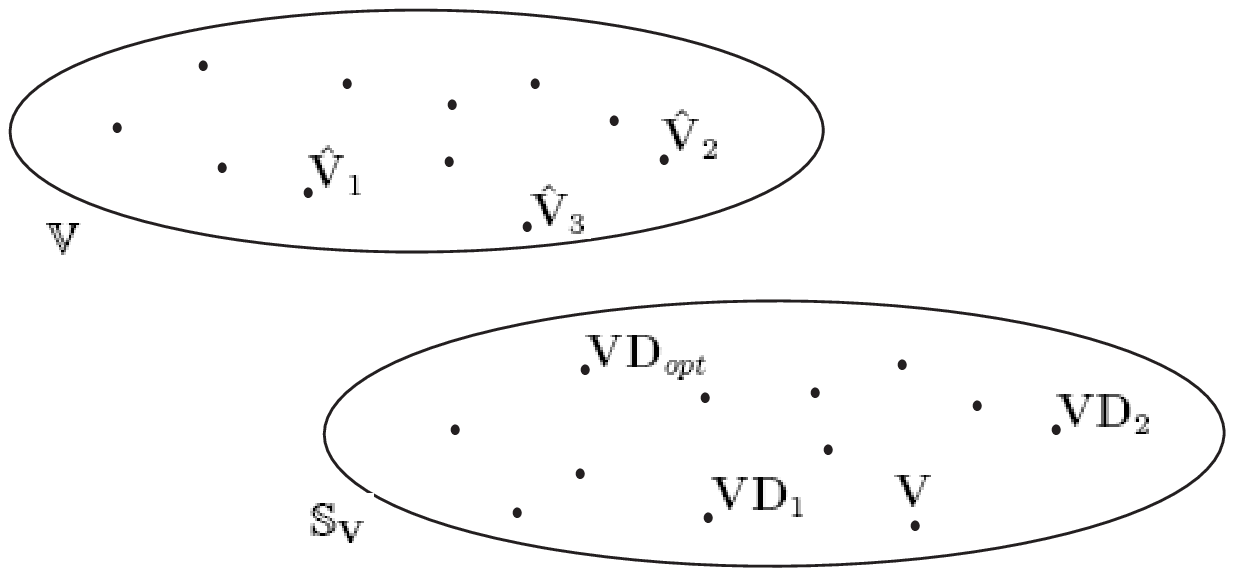}
\caption{Sets illustrating the the codebook elements,$\mathbb{V}$,
and unitary matrices from SVD, $\mathbb{S}_{\mathbf{V}}$.}
\label{fig:BICMB_sets}
\end{figure}

\emph{ Proposition 1}: The minimization in (\ref{eq:Eclidean_VD})
is equivalent to the following minimization problem
\begin{align}
\mathbf{V}_L=\underset{\mathbf{\hat{V}}_i \in
\mathbb{V}}{\arg\min} \mbox{ }|| \mbox{
}\mathbf{VD}^{opt}-\mathbf{\hat{V}}_i\mbox{ }||_F^2.
\label{eq:Eclidean_Vopt}
\end{align}
The $k^{th}$ diagonal element of the diagonal matrix
$\mathbf{D}^{opt}$ is given as
\begin{align}
\theta_k^{opt}=-\phi_k \quad \quad{ }k=1,2,\ldots,S
\label{eq:Eclidean_Vopt2}
\end{align}
where $0\leq\phi_k<2\pi$ is the phase of
$\mathbf{\hat{v}}_{ik}^H\mathbf{v}_k$ and where the vectors
$\mathbf{\hat{v}}_{ik}$ and $\mathbf{v}_k$ correspond to the
$k^{th}$ column of $\mathbf{\hat{V}}_i$ and $\mathbf{V}$,
respectively.

\begin{proof}
Without loss of generality, let $N\leq M$ and $S=N$ streams be
used. For the other cases, the matrices are replaced by their
first $S$ columns. The term to be minimized in
(\ref{eq:Eclidean_VD}) can be expressed as
\begin{align}
||\mbox{ }\mathbf{VD}-\mathbf{\hat{V}}_i \mbox{ } ||_F^2 = \mbox{
}&2\textrm{tr}[\mathbf{I}]-\textrm{tr}\left[ \mathbf{\hat{V}}_i^H
\mathbf{VD} + (\mathbf{\hat{V}}_i^H
\mathbf{VD})^H\right]\label{eq:proof_frobenius_norm1}\\
=\mbox{ }&2N-2\textrm{tr}\left[\Re [\mathbf{\hat{V}}_i^H
\mathbf{VD}]\right]\nonumber\\
=\mbox{ }& 2N-2\Re \left[\sum\limits_{k=1}^
{N}\mathbf{\hat{v}}_{ik}^H\mathbf{v}_{k}e^{j\theta_k}\right]
\label{eq:proof_frobenius_norm2}
\end{align}
where $\mathbf{D} =$ diag$(e^{j\theta_1}, e^{j\theta_2}, \ldots,
e^{j\theta_N})$, $\mathbf{\hat{v}}_{ik}$ and $\mathbf{v}_k$
correspond to the $k^{th}$ column of $\mathbf{\hat{V}}_i$ and
$\mathbf{V}$, respectively. Minimizing
(\ref{eq:proof_frobenius_norm1}) is equivalent to maximizing the
second term in (\ref{eq:proof_frobenius_norm2}). It is easy to see
that the optimal value of $\theta_k$  maximizing the summation in
(\ref{eq:proof_frobenius_norm2}) is
\begin{align}
\theta_k^{opt}=-\phi_k \quad \quad{  }k=1,2,\ldots,N
\label{eq:proof6}
\end{align}
where $0\leq\phi_k<2\pi$ is the phase of
$\mathbf{\hat{v}}_{ik}^H\mathbf{v}_k=|\mathbf{\hat{v}}_{ik}^H\mathbf{v}_k|e^{j\phi_k}$.
\end{proof}

Proposition 1 results in the following optimal selection criterion
in the Euclidean sense.\\

\emph{Selection Criterion - Optimal Euclidean (SC-OE)} :  The
receiver selects $\mathbf{V}_L$ such that
\begin{align}
\mathbf{V}_L=\underset{\mathbf{\hat{V}}_i \in
\mathbb{V}}{\arg\min} \mbox{ }|| \mbox{
}\mathbf{VD}^{opt}-\mathbf{\hat{V}}_i\mbox{ }||_F^2.
\label{eq:Eclidean_Voptimal}
\end{align}

Note that, in (\ref{eq:Eclidean_Voptimal}), $\mathbf{D}^{opt}$
depends on both $\mathbf{V}$ and $\mathbf{\hat{V}}_i$. Employing
(\ref{eq:Eclidean_Voptimal}), one can apply the well-known
generalized Lloyd algorithm \cite{lindeLGB} to design an optimum
codebook $\mathbb{V}$. The resulting codebook can then be used
together with (\ref{eq:Eclidean_Voptimal}), as a limited-rate CSIT
BICMB system. To that end, we will need centroids for the
generalized Lloyd algorithm. We will calculate these new centroids
in the next subsection.

\subsection{Codebook Design}\label{sec:Codebook_design}
Our codebook design is based on generalized Lloyd's algorithm
\cite{lindeLGB}. We will minimize the average distortion
\begin{align}
J= E\left[\underset{\mathbf{\hat{V}}_i \in \mathbb{V}}{\min}\mbox{
} || \mbox{ }\mathbf{VD}^{opt}-\mathbf{\hat{V}}_i\mbox{
}||_F^2\right] \label{eq:VQ_distortion}.
\end{align}

Here, the distortion measure we intend to use is
\begin{align}
d(\mathbf{\hat{V}}_i,\mathbf{V}) = ||\mbox{
}\mathbf{V}-\mathbf{\hat{V}}_i\mbox{ }||_F^2.
\label{eq:VQ_distortion2}
\end{align}
But, due to the previous discussion, we need to calculate the
distortion between each $\mathbf{\hat{V}}_i$ and the whole set
$\mathbb{S}_{\mathbf{V}}$. As a result, we employ
\begin{align}
d_1(\mathbf{\hat{V}}_i,\mathbf{V}) = || \mbox{
}\mathbf{VD}^{opt}-\mathbf{\hat{V}}_i\mbox{ }||_F^2
\label{eq:VQ_distortion3}
\end{align}
due to the nonuniqueness property of SVD. We assume that $B$ bits
are reserved for the limited feedback link to quantize the optimal
beamforming matrix. In this algorithm, we will begin with an
initial codebook of matrices
$\mathbb{\tilde{V}}_{0}=\{{\mathbf{\tilde{V}}_{0,k}}\}_{k=1}^{2^B}$
and iteratively improve it to generate a set of matrices
$\mathbb{\tilde{V}}_{m}=\{{\mathbf{\tilde{V}}_{m,k}}\}_{k=1}^{2^B}$
until the algorithm converges. The algorithm can be summarized by
the following steps:

1) Generate a large training set of channel matrices,
$\mathbf{H}(n)$ and their corresponding right singular matrices
$\mathbf{V}(n)$. Let $\Psi$ be the set of all $\mathbf{V}(n)$s.

2) Generate an initial codebook of unitary matrices,
$\mathbb{\tilde{V}}_{0}=\{{\mathbf{\tilde{V}}_{0,k}}\}_{k=1}^{2^B}$.

3) Set $m=1$.

4) Partition the set of training matrices into $P=2^B$
quantization regions where the $k^{th}$ region is defined as
\begin{align}
\mathbb{X}_k = \{\mathbf{V} \in \Psi |\mbox{ } ||\mbox{
}\mathbf{VD}^{opt}-\mathbf{\tilde{V}}_{m-1,k}\mbox{ }||_F^2 \leq
||\mbox{ }\mathbf{VD}^{opt}-\mathbf{\tilde{V}}_{m-1,l}\mbox{
}||_F^2\mbox{ }\mbox{ } \forall \mbox{ }k\neq
l\}\label{eq:VQ_partition}
\end{align}

5) Using the given partitions, construct a new codebook
$\tilde{\mathbb{V}}_m$, with the $k^{th}$ beamforming matrix being
\begin{align}
\mathbf{\tilde{V}}_{m,k}=  \underset{\mathbf{\hat{V}} : \mbox{ }
\mathbf {\hat{V}}^H\mathbf{\hat{V}=I}}{\arg\min} E \left[ ||\mbox{
}\mathbf{VD}^{opt}-\mathbf{\hat{V}} \mbox{ } ||_F^2 \mbox{ }|
\mbox{ } \mathbf{V} \in \mathbb{X}_k\right].
\label{eq:codebook_opt}
\end{align}

6) Define
\begin{align}
J_m = \sum\limits_{i=1}^{2^B}\sum
\limits_{n:\mathbf{V}(n)\rightarrow\mathbf{\tilde{V}}_{m,i}}||\mbox{
}\mathbf{V}(n)\mathbf{D}^{opt}-\mathbf{\tilde{V}}_{m,i}\mbox{
}||_F^2 \label{eq:VQ_ave_distortion_m}
\end{align}
where $\mathbf{V}(n)\rightarrow\mathbf{\tilde{V}}_{m,i}$ means
$\mathbf{\tilde{V}}_{m,i}=\underset{\mathbf{\hat{V}}_j\in\tilde{\mathbb{V}}_m}{\arg\min}\mbox{
}d_1(\mathbf{\hat{V}}_j,\mathbf{V}(n))$ . If
$(J_{m-1}-J_m)/J_{m-1}>\epsilon$, set $m=m+1$ and go back to Step
4. Otherwise, terminate the algorithm and set the codebook
$\mathbb{V}= \tilde{\mathbb{V}}_m$.

The optimal solution of the optimization problem in
(\ref{eq:codebook_opt}) gives the optimal centroid for the
corresponding region. The distortion measure to be minimized can
be rewritten as
\begin{align}
||\mbox{ }\mathbf{VD}^{opt}-\mathbf{\hat{V}}_i\mbox{ }||_F^2
=\mbox{ }&2N-2\textrm{tr}\left[\Re [\mathbf{\hat{V}}^H
\mathbf{VD}^{opt}]\right]\nonumber\\
=\mbox{ }& 2N-2\Re \left[\sum\limits_{s=1}^
{N}\mathbf{\hat{v}}_s^H\mathbf{v}_se^{j\theta_s^{opt}}\right]\nonumber\\
=\mbox{ }& 2N- 2\sum\limits_{s=1}^
{N}|\mathbf{\hat{v}}_s^H\mathbf{v}_s|\label{eq:VQ_frobenius_norm}
\end{align}
where (\ref{eq:VQ_frobenius_norm}) follows by using the optimal
$\theta_s^{opt}$ previously derived in (\ref{eq:proof6}).
Therefore the original optimization problem in
(\ref{eq:codebook_opt}) can be rewritten as
\begin{align}
\mathbf{\tilde{V}}_{m,k}=  \underset{\mathbf{\hat{V}} : \mbox{ }
\mathbf{\hat{V}}^H\mathbf{\hat{V}=I}}{\arg\max} E \left[
\sum\limits_{s=1}^ {N}|\mathbf{\hat{v}}_s^H\mathbf{v}_s|\mbox{ } |
\mbox{ } \mathbf{V} \in \mathbb{X}_k\right].
\label{eq:codebook_opt2}
\end{align}

The maximization problem above does not have a tractable
analytical solution. Next, we will modify the problem to find an
approximate analytical solution. Note that the expectation in
(\ref{eq:codebook_opt2}) can be written as the sum of expectation
of each term due to the linearity of the expectation operation. We
will relax the unitary constraint on $\mathbf{\hat{V}}$ and
replace the constraint with having unit norm columns. In this
case, the modified optimization problem is equivalent to finding
independent optimal vectors which maximize each expectation in
(\ref{eq:codebook_opt2}). The individual maximization problem
becomes
\begin{align}
\mathbf{\tilde{e}}_{m,k}^{(i)}=  \underset{\mathbf{\hat{e}} :
\mbox{ } ||\mathbf{\hat{e}}||_2^2=1}{\arg\max} \mbox{ }E \left[
|\mathbf{\hat{e}}^H\mathbf{v}_i|\mbox{ } | \mbox{ } \mathbf{v}_i
\in \mathbb{X}_k^{(i)}\right] \quad
i=1,2,\ldots,N\label{eq:codebook_vec_opt}
\end{align}
where $\mathbb{X}_k^{(i)}$ corresponds to the space of the
$i^{th}$ column of the elements in $\mathbb{X}_k$. The optimal
solution for (\ref{eq:codebook_vec_opt}) is \cite{rohVQ_SB}
\begin{align}
\mathbf{\tilde{e}}_{m,k}^{(i)}= \mbox{ principal eigenvector of }
E \left[\mathbf{v}_i \mathbf{v}_i^H|\mbox{ } \mathbf{v}_i \in
\mathbb{X}_k^{(i)}\right] \label{eq:codebook_vec_opt_solution}
\end{align}
where the numerical averaging over $\mathbb{X}_k^{(i)}$ is
substituted for expectation during codebook design. Let
$\mathbf{\tilde{E}}_{m,k}$ be the matrix whose columns are found
from (\ref{eq:codebook_vec_opt_solution}), maximizing the
expectation in (\ref{eq:codebook_vec_opt}) and approximating the
maximization in (\ref{eq:codebook_opt2}). Note that this matrix is
not necessarily unitary, therefore to find the centroid we will
utilize Euclidean projection to find the closest unitary matrix as
follows
\begin{align}
\mathbf{\tilde{V}}_{m,k}=\underset{\mathbf{\hat{V}} : \mbox{ }
\mathbf{\hat{V}}^H\mathbf{\hat{V}=I}}{\arg\min} ||\mbox{
}\mathbf{\tilde{E}}_{m,k}-\mathbf{\hat{V}}\mbox{ }||_F^2.
\label{eq:codebook_closest_unitary}
\end{align}
The closest unitary matrix can be found in closed form as
\cite{horn}
\begin{align}
\mathbf{\tilde{V}}_{m,k}=\mathbf{\tilde{U}}\mathbf{\tilde{W}}^H
\label{eq:codebook_closest_unitary2}
\end{align}
where
$\mathbf{\tilde{E}}_{m,k}=\mathbf{\tilde{U}\tilde{\Sigma}}\mathbf{\tilde{W}}^H$.

The approach explained above to find the centroid in each region
reduces to the optimal solution for the single beamforming case.
Although it may be suboptimal for the multiple beamforming case,
the centroid found from (\ref{eq:codebook_closest_unitary2})
enables the algorithm to have monotonic decrease in average
distortion given by (\ref{eq:VQ_distortion}) in each iteration and
to converge to a local minimum, as shown in Figure
\ref{fig:BICMB_average_dist} for a $2\times2$ scenario with 2
streams and 4-bit feedback.

\begin{figure}[pg6]
\centering \includegraphics[width =0.7
\linewidth]{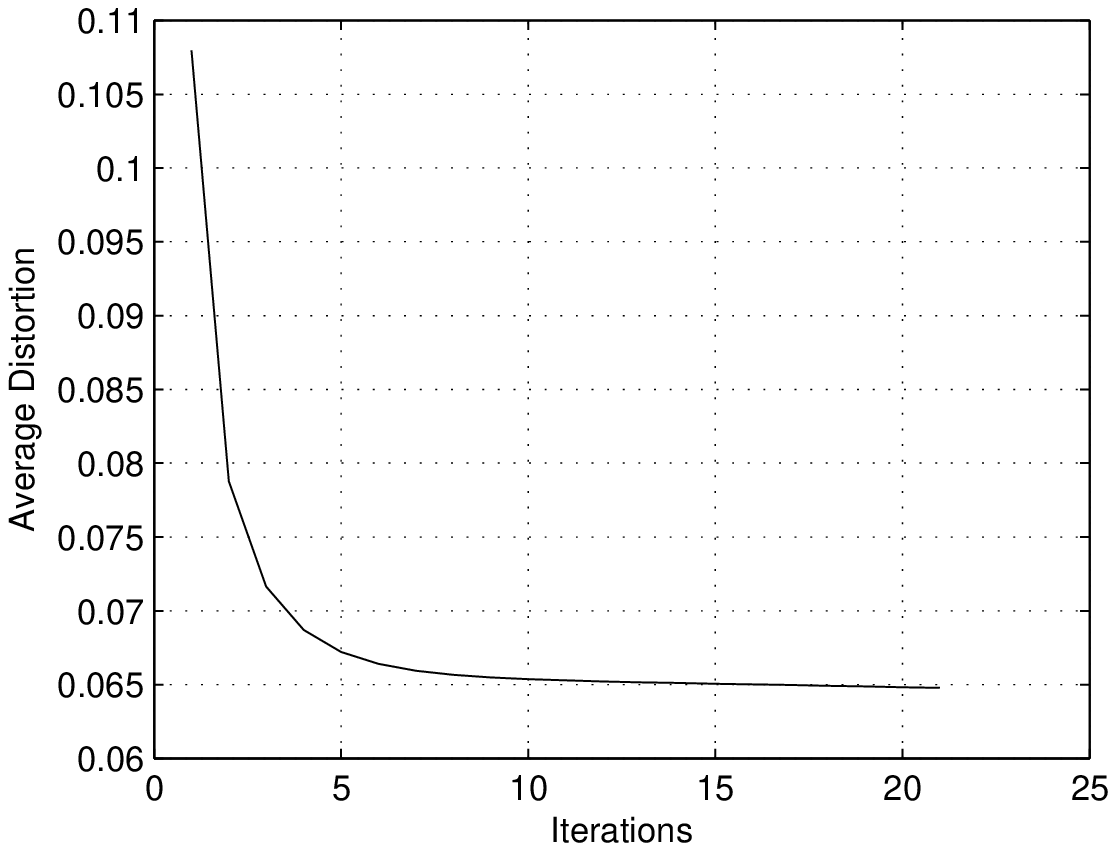} \caption{Average
distortion for $2\times2$ system with 4-bit feedback.}
\label{fig:BICMB_average_dist}
\end{figure}

\subsection{The Receiver}\label{sec:Detectors_BICMB_LF}

We will first discuss the ZF and MMSE receivers and then
describe a receiver based on SVD. We show in the appendix that the performance of these three decoders is the same when there is perfect CSIT.\\

1) \emph{ZF Receiver}

When there is only limited feedback for the quantization of
$\mathbf{V}$, i.e., $\mathbf{V}_L$ is used as the precoder, the
diagonalization of the channel will be lost and with the ZF
detector, the system input-output relation becomes;
\begin{align}
\mathbf{r} =\mathbf{Gy} = \mathbf{x}+\mathbf{Gn}
\label{eq:InOutBeamLFZF}
\end{align}
where $\mathbf{G}=
(\mathbf{HV}_L)^\dag=[(\mathbf{HV}_L)^H(\mathbf{HV}_L)]^{-1}(\mathbf{HV}_L)^H$.
In this case, we will use the following bit metrics
\cite{butler_ZFbitmetrics},
\begin{align}
\gamma^i (r_{k,s},c_{k'}) =  \min\limits_{x \in \chi_{c_{k'}}^i}
\frac{|r_{k,s} - x|^2}{{||\mathbf{g}_s||}^2}
\label{eq:BitMetricsZF}
\end{align}
where $r_{k,s}$ is the received signal after equalization at time
$k$ on the $s^{th}$ stream and $\mathbf{g}_s$ is the $s^{th}$
column of $\mathbf{G}^T$.

In the perfect CSIT case, where the channel right singular matrix
$\mathbf{V}$ is perfectly known at the transmitter, the bit
metrics (\ref{eq:BitMetricsZF}) of the ZF receiver are equal to
that of the optimum BICMB receiver. The proof is provided in the
Appendix.\\

 2) \emph{MMSE Receiver}

MMSE detector is a superior solution to the linear detection problem
which balances ISI against noise enhancement. The corresponding
input-output relation is given by (\ref{eq:InOutBeamLFZF}), where
now $\mathbf{G}$ is given by
\begin{align}
\mathbf{G} =
[\mathbf{(HV}_L)^H\mathbf{HV}_L+\sigma^2\mathbf{I}]^{-1}\mathbf{(HV}_L)^H
\label{eq:LFMMSE_detector}
\end{align}
and where $\sigma^2=N/SNR$ from the system model given in Section
\ref{sec:QBICMBsystem}. We will use the following bit metrics
\cite{seethalerMMSE}
\begin{align}
\gamma^i (r_{k,s},c_{k'}) =  \min\limits_{x \in \chi_{c_{k'}}^i}
\frac{{W_{ss}}}{1-{W_{ss}}}\left|\frac{r_{k,s}}{{W_{ss}}}-x\right|^2
\label{eq:BitMetricsMMSE}
\end{align}
where
$\mathbf{W}=[\mathbf{I}+\sigma^2[\mathbf{(HV_L})^H\mathbf{HV_L}]^{-1}]^{-1}$
and $W_{ss}$ is the $s^{th}$ diagonal element of $\mathbf{W}$.

In the perfect CSIT case, the MMSE receiver
is equivalent to the optimum BICMB receiver. The proof is
provided in the Appendix.\\

3) \emph{SVD Receiver}

In the case of perfect knowledge of $\mathbf{V}$ at the
transmitter, the receiver can use the $\mathbf{U}^H$ matrix to
diagonalize the channel, where $\mathbf{H = U\Sigma V}^H$. In the
case of limited feedback, the $\mathbf{U}^H$ matrix can still be
used as an equalizer \cite{sadrabadiGivens},
\cite{mielczarek_VQ05}. In this section, we will provide a linear detector which performs the same as the $\mathbf{U}^H$ detector with lower complexity. Note that, we proposed an optimum
selection criterion in (\ref{eq:Eclidean_Voptimal}) which is
needed because of the nonuniqueness property of SVD. The optimized
selection criterion aims to quantize $\mathbf{VD}^{opt}$ instead
of $\mathbf{V}$. Each element of the diagonal unitary matrix
$\mathbf{D}^{opt}$ can be found from (\ref{eq:proof6}) and it is
dependent on the codebook elements and the instantaneous channel
realization. From Theorem 1, it is easy to see that there is a
unique matching left singular matrix for $\mathbf{VD}^{opt}$,
which can be used as a detector. Therefore the corresponding
linear equalizer matrix is
\begin{align}
\mathbf{G} = (\mathbf{UD}^{opt})^H.\label{eq:proposed_detector}
\end{align}

In this case, when $\mathbf{V}_L$ is used as a precoder at the
transmitter, the baseband system input-output relation is
\begin{align}
\mathbf{r} &= \mathbf{GHV}_L\mathbf{x}+\mathbf{Gn}\\
&=(\mathbf{D}^{opt})^H\mathbf{\Sigma
V}^H\mathbf{V}_L\mathbf{x}+\mathbf{n'} \label{eq:InOutBeamLF_U}
\end{align}
where in (\ref{eq:InOutBeamLF_U}) $\mathbf{H}$ is replaced by its
SVD. Note that because $\mathbf{G}$ is a unitary transformation the
noise vectors $\mathbf{n'}$ and $\mathbf{n}$ have the same
statistics. Then the input-output relation for the $s^{th}$ stream
becomes
\begin{align}
r_s &=\lambda_se^{-j\theta_s^{opt}}\sum\limits_{i=1}^S \mbox{ }
\mathbf{v}_s^H
\mathbf{v}_{L,i}x_i+n'_s\\
&=\lambda_se^{-j\theta_s^{opt}}\mathbf{v}_s^H\mathbf{v}_{L,s}x_s+
\lambda_se^{-j\theta_s^{opt}}\sum\limits_{i=1, \mbox{ }i \neq s}^S
\mbox{ } \mathbf{v}_s^H \mathbf{v}_{L,i}x_i +
n'_s \\
&=\lambda_s|\mathbf{v}_s^H\mathbf{v}_{L,s}|x_s+
\lambda_se^{-j\theta_s^{opt}}\sum\limits_{i=1, \mbox{ }i \neq s}^S
\mbox{ } \mathbf{v}_s^H \mathbf{v}_{L,i}x_i + n'_s.
\label{eq:InOutBeamLF_UTheta}
\end{align}
Note that the first term has the desired signal, the second term is
interference from other streams, and the third term is noise. The
transmitted symbols $x_i$ are typically from symmetric
constellations. Therefore, the mean of $x_i$ is zero. As discussed
previously, we normalize its variance to 1. Due to bit interleaving,
$x_i$, $i=1,2,\ldots,S$ are uncorrelated. For a given channel
realization, (\ref{eq:InOutBeamLF_UTheta}) can be written in a
compact form as
\begin{align}
r_s &=\tilde{\lambda}_sx_s+\tilde{n}_s \label{eq:compact_ReceiverU}
\end{align}
where $\tilde{\lambda}_s=\lambda_s|\mathbf{v}_s^H\mathbf{v}_{L,s}|$
and $\tilde{n}_s$ is approximated as a zero-mean complex Gaussian
random variable with variance
$\tilde{\sigma}_s^2=\lambda_s^2\sum_{i=1,i\neq
s}^{S}|\mathbf{v}_s^H\mathbf{v}_{L,i}|^2+N/SNR$. We determined
through simulations that the Gaussian approximation is highly
accurate for low and intermediate SNR values (e.g., 15 dB) or when
the number of feedback bits is beyond 4. Although for large SNR
(e.g., 30 dB), the approximation is less accurate, as the feedback
rate increases, the accuracy loss diminishes independent of SNR. In
addition, this approximation enables a very simple bit-metric
calculation similar to the perfect CSIT case.

Let $\chi_b^i$ denote the subset of all signals $x \in \chi$ whose
label has the value $b \in \{0,1\} $ in position $i$. The bit
metrics for (\ref{eq:compact_ReceiverU}) are given by \cite{bicm}
\begin{align}
\gamma^i (r_{k,s},c_{k'}) =  \min\limits_{x \in \chi_{c_{k'}}^i}
\frac{|r_{k,s}-\tilde{\lambda}_sx|^2}{\tilde{\sigma}_s^2}
\label{eq:BitMetricsUheta}.
\end{align}

In the sequel, we will call the receiver proposed in this section
as the SVD receiver. We will show in the next section that the
performance of the SVD receiver is close to that of the MMSE
receiver for the $2\times 2$ MIMO system. The advantage of the SVD
receiver over the MMSE receiver is its relative simplicity since
it avoids the matrix inversions needed in
(\ref{eq:LFMMSE_detector}) and (\ref{eq:BitMetricsMMSE}). One can
observe that when the limited feedback rate is low, the
interference term may dominate the noise term, which may result in
poor performance. We emphasize that the optimum receiver with a linear detector
for the limited-rate CSIT system described in the previous section
is the MMSE receiver. However, the SVD receiver is a simpler one
with a performance tradeoff against the MMSE receiver
while consistently outperforming the ZF receiver.\\

\section{Simulation Results} \label{sec:Results}

In the simulations below, the industry standard 64-state 1/2-rate
(133,171) $d_{free}=10$ convolutional code is used and the
constellation is 16-QAM. As in all similar work, the channel is
assumed to be quasi-static and flat fading.

Figure \ref{fig:BICMB_LF_figure1} illustrates that in the case of
BICMB, the precoder matrix $\mathbf{V}$ is not invariant under a
general unitary matrix transformation. As discussed previously,
assumption of this invariance results in the Grassmannian codebook
design approach studied widely in the literature
\cite{love03LF-SM}. Again, as discussed
previously, most of the work in the literature is for uncoded
systems where invariance under a general unitary matrix
transformation follows from the use of optimization criterion such
as MSE, SNR, or mutual information. All curves in this figure
employ BICMB with ZF receiver, while the solid ones employ the
$\mathbf{V}$ matrix given by SVD of $\mathbf{H}$, those with
broken lines employ $\mathbf{V}'=\mathbf{VQ}$ where $\mathbf{Q}$
is a $2\times2$ DFT matrix, which is unitary. Clearly, BICMB
performance is not invariant under a general unitary matrix
transformation.

Figure \ref{fig:BICMB_LF_figure2} shows a number of different
systems to illustrate the improvement due to the new selection
criterion (\ref{eq:Eclidean_Voptimal}). This selection criterion is
compared to the one that maximizes the minimum eigenvalue
($\lambda_{min}$) of $\mathbf{HV}_i$. This method is employed in
\cite{love03LF-SM} with the ZF receiver. In order to show that there
is a gain due to (\ref{eq:Eclidean_Voptimal}), we use the ZF
receiver in our system as well. Systematic generation of codebooks
\cite{hochwald99} with a selection criterion that maximizes
$\lambda_{min}$ is used for the curves with legend
SC-$\lambda_{min}$ and the randomly generated codebook with the
selection criterion in (\ref{eq:Eclidean_Vopt}) and
(\ref{eq:Eclidean_Vopt2}) is used for the curves with legend SC-OE.
As can be seen, the performance is improved significantly with the
proposed approach, and with $\lambda_{min}$ approach, the
performance saturates with increasing the number of bits.

Figure \ref{fig:BICMB_LF_figure3} compares (\ref{eq:Eclidean_V})
and (\ref{eq:Eclidean_Voptimal}) employing two receiver
structures: ZF and MMSE. The codebook employed is randomly
generated. There is clearly a significant gain due to
(\ref{eq:Eclidean_Voptimal}) for both receivers. In Figure
\ref{fig:BICMB_LF_figure4} the performance of the SVD receiver is
compared with the ZF and MMSE receivers for the $2\times 2$
scenario with 2 streams. All curves in the figures use the
optimized Euclidean criterion with a randomly generated codebook.
The SVD receiver, which exploits the nonuniqueness of SVD both at
the transmitter and the receiver, significantly outperforms the ZF
receiver and achieves a performance very close to the MMSE
receiver for the 8-bit scenario. Note that, the overall complexity
of the system with the SVD receiver is less than the one with the
MMSE receiver. When the number of feedback bits is 8 for the
$2\times2$ case, it achieves a performance 0.25 dB close to the
unquantized system.

Figure \ref{fig:BICMB_LF_figure8} shows the simulation results for
various receivers in a $2\times 2$ system when the codebook is
designed using the VQ algorithm discussed in Section
\ref{sec:Codebook_design}. All curves use the optimal Euclidean
criterion. As seen from the figure, the performance of the
randomly generated codebook (RVQ) can be significantly improved
for all receivers. To illustrate, the performance of MMSE 8-bit
RVQ and 6-bit VQ are very close to each other, therefore 2 bit
reduction is achieved via the proposed codebook design. A similar
reduction can be observed for the SVD receiver. On the other hand,
for the same number of feedback bits, 2 dB performance gain is
achievable for the ZF receiver. Note that there is significant
performance degradation when the ZF receiver is used for both RVQ
and VQ scenarios compared to
the MMSE and SVD receivers.\\

\begin{figure}[pg19]
\centering \includegraphics[width = 0.75
\linewidth]{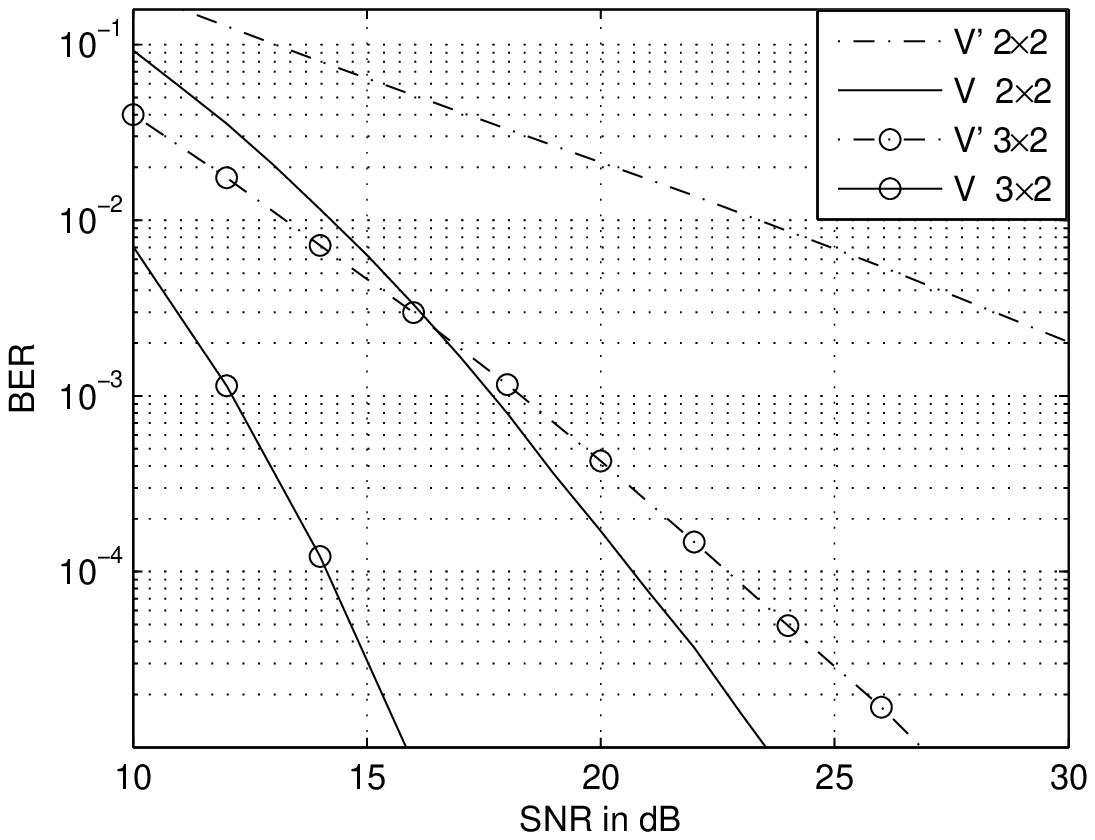} \caption{The
precoding matrix given by SVD $\mathbf{V}$ vs.
$\mathbf{V}'=\mathbf{VQ}$ where $\mathbf{Q}$ is a DFT matrix $2
\times 2$ and $3 \times 2$ system with 2 streams.}
\label{fig:BICMB_LF_figure1}
\end{figure}

\begin{figure}[m]
\centering \includegraphics[width = 0.75
\linewidth]{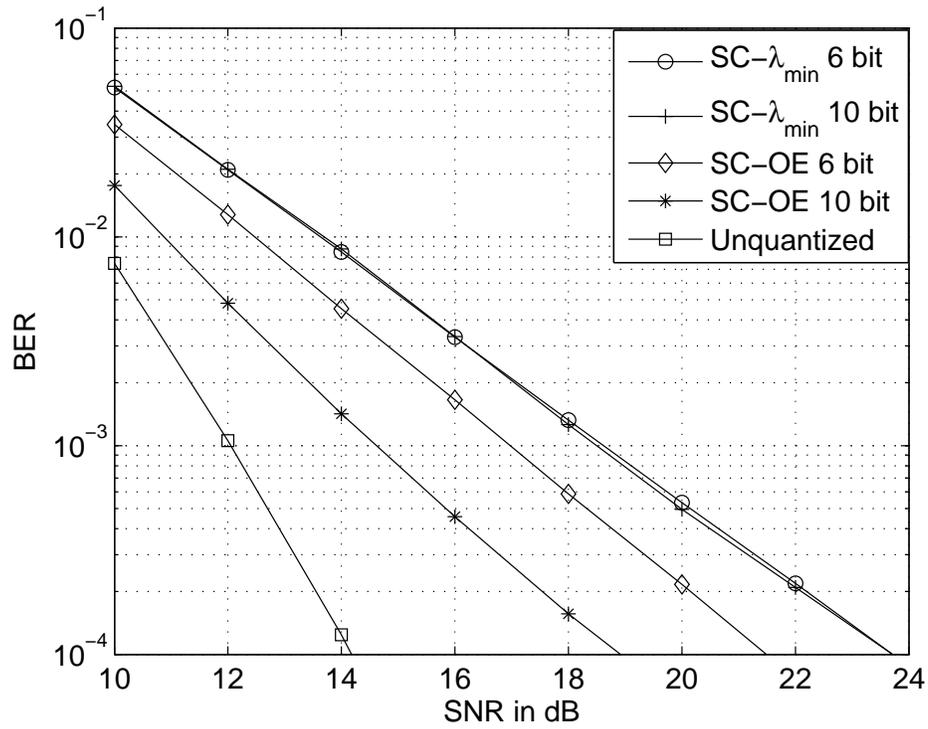}
\caption{New selection criterion vs. $\lambda_{min}$-based
selection criterion $3 \times 2$ system with 2 streams.}
\label{fig:BICMB_LF_figure2}
\end{figure}

\begin{figure}[m]
\centering \includegraphics[width = 0.75
\linewidth]{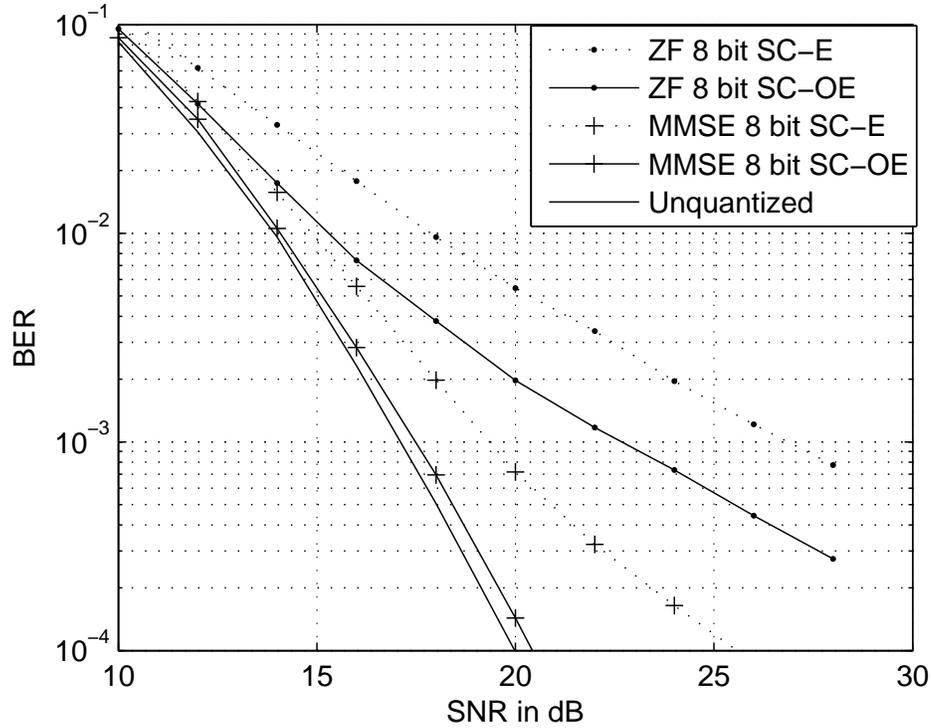}
\caption{Comparison SC-E (\ref{eq:Eclidean_V}) and SC-OE
(\ref{eq:Eclidean_Voptimal}) $2 \times 2$ system with 2 streams.}
\label{fig:BICMB_LF_figure3}
\end{figure}

\begin{figure}[m]
\centering \includegraphics[width = 0.75
\linewidth]{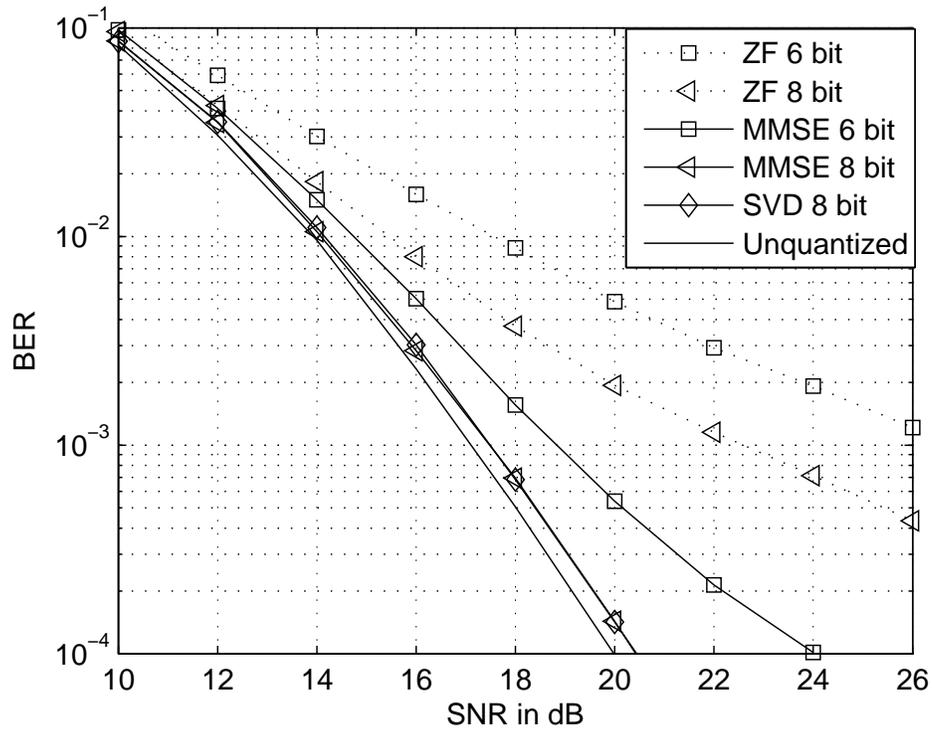} \caption{The SVD
receiver vs. MMSE and ZF receivers using SC-OE with randomly
generated codebook $2 \times 2$ system with 2 streams.}
\label{fig:BICMB_LF_figure4}
\end{figure}

\begin{figure}[m]
\centering \includegraphics[width =
0.75\linewidth]{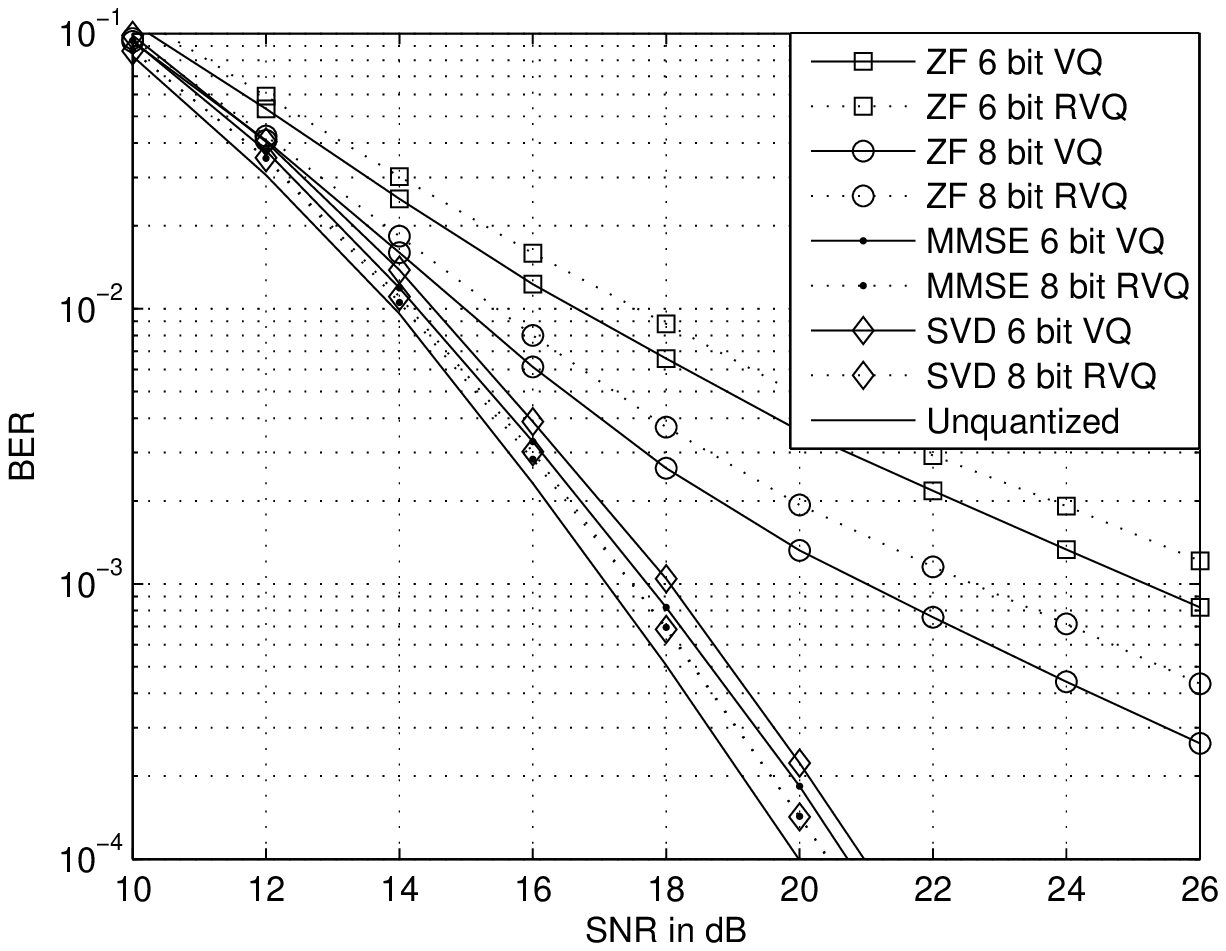} \caption{The
SVD receiver vs. MMSE and ZF receivers with RVQ and VQ using SC-OE
$2 \times 2$ system with 2 streams.} \label{fig:BICMB_LF_figure8}
\end{figure}

\section{Conclusion}\label{sec:concl}

BICMB is a high-performance and low-complexity broadband wireless
system with full spatial multiplexing and full diversity. However,
the system requires perfect knowledge of the channel right
singular vectors, which is not practical in a real environment.

This paper addressed the performance of BICMB with limited CSIT
feedback using a codebook-based approach. We proposed a new
optimal distortion measure for selecting the best precoder from a
given codebook. The centroids for this distortion measure are
calculated. Codebook design is performed via the generalized Lloyd
algorithm based on the new distortion measure and the new
centroids. We provided simulation results demonstrating the
performance improvement achieved with the proposed distortion
measure for various receivers with linear detectors.\\

\appendix
\section{Appendix}\label{sec:appn}

In the perfect CSIT case, the transmitter uses the right
singular matrix $\mathbf{V}$ as the precoding matrix
$\mathbf{V}_L$. The $N \times S$ precoding matrix can be expressed
as $\mathbf{V}_L = \mathbf{V \Phi}_N$, where the $N \times S$
matrix $\mathbf{\Phi}_N$ is used to select the first $S$ columns
of $\mathbf{V}$, defined as
\begin{displaymath}
\mathbf{\Phi}_N = \left[ \begin{array}{ccc} \mathbf{I}_{S}\\
\hline \mathbf{0}_{N-S,S} \end{array} \right],
\end{displaymath}
and $\mathbf{0}_{N-S,S}$ is an $(N-S) \times S$ matrix whose
elements are all zeros. Therefore, the system input-output
relation in (\ref{eq:InOutBeamLF}) can be written as
\begin{align}
\mathbf{y} &= \mathbf{U \hat{\Sigma} x}+\mathbf{n},
\label{eq:InOutBeamPerfect}
\end{align}
where $\mathbf{\hat{\Sigma}}$ is defined as
\begin{displaymath}
\mathbf{\hat{\Sigma}} = \mathbf{\Sigma \Phi}_N = \left[ \begin{array}{ccc} \mathbf{\Sigma}_{S} \\
\hline \mathbf{0}_{M-S,S} \end{array} \right],
\end{displaymath}
and $\mathbf{\Sigma}_{S}$ is an $S \times S$ square matrix whose
elements are taken from the largest $S$ singular values of $\mathbf{H}$.\\

1) \emph{BICMB Receiver}

The optimum detector for the BICMB receiver is the corresponding
left singular matrix $\mathbf{U}^H$. The baseband input-output
relation for each subchannel becomes \cite{akayTC06BICMB}
\begin{align}
r'_{k,s} = \lambda_s x_{k,s} + n_{k,s}
\label{eq:InOutBeam_BICMB-U}
\end{align}
for $s=1, 2, \ldots, S$ where $\lambda_s$ is the $s^{th}$ channel
singular value and $r'_{k,s}$ is the detected symbol of the
$s^{th}$ subchannel at the $k^{th}$ time instant which is defined
as in (\ref{eq:MLrule}). Then, the following ML bit metrics for
the BICMB soft input Viterbi decoder are used \cite{akayWCNC06},
\cite{akayTC06BICMB}
\begin{equation}
\gamma_{BICMB}^i (r'_{k,s},c_{k'}) =  \min\limits_{x \in
\chi_{c_{k'}}^i} |r'_{k,s} - \lambda_s x|^2
\label{eq:BitMetrics}
\end{equation}
where $k'$ is defined as in (\ref{eq:MLrule}).\\

2) \emph{ZF Receiver}

After the ZF detector, the system input-output relation becomes
\begin{align}
\mathbf{r}_{ZF} =\mathbf{\hat{\Sigma}}^{\dag} \mathbf{U}^H
\mathbf{y} = \mathbf{x} + \mathbf{\hat{\Sigma}}^{\dag}
\mathbf{U}^H \mathbf{n} \label{eq:InOutBeamPerfectZF}
\end{align}
where $\mathbf{G}$\ in (\ref{eq:InOutBeamLFZF}) is replaced by
($\mathbf{HV}_L)^{\dag}$ = ($\mathbf{U \Sigma V}^H \mathbf{V
\Phi}_N)^{\dag}$ = $\mathbf{\hat{\Sigma}}^{\dag} \mathbf{U}^H$.
Note that the last equality $(\mathbf{AB})^{\dag} =
\mathbf{B}^\dag \mathbf{A}^\dag$ holds if $\mathbf{A}^H \mathbf{A}
= \mathbf{I}$ \cite{greville66SIAM}. Accordingly, the baseband
input-output for each substream becomes
\begin{align}
\hat{r}_{k,s} =  x_{k,s} + \lambda_s^{-1}n_{k,s} =
\lambda_s^{-1}r'_{k,s}, \label{eq:InOutBeam_Perfect-ZF}
\end{align}
where the relation with $r'_{k,s}$ is obvious when
(\ref{eq:InOutBeam_Perfect-ZF}) is compared with
(\ref{eq:InOutBeam_BICMB-U}).

To calculate the $s^{th}$ column of $\mathbf{G}^T$ for metric
calculation in (\ref{eq:BitMetricsZF}), consider $\mathbf{U}
=(\mathbf{u}_1 \vdots \mathbf{u}_2 \vdots \ldots \vdots
\mathbf{u}_M$), where $\mathbf{u}_1, \mathbf{u}_2, \ldots,
\mathbf{u}_M$ are the column vectors of $\mathbf{U}$. Then,
\begin{align}
\mathbf{G}^T = (\mathbf{\hat{\Sigma}}^{\dag} \mathbf{U}^H )^T =
(\mathbf{u}_1^* \vdots \mathbf{u}_2^* \vdots \ldots \vdots \mathbf{u}_M^*) \left[ \begin{array}{ccc} \mathbf{\Sigma}_{S}^{-1} \\
\hline \mathbf{0}_{M-S,S} \end{array} \right].
\end{align}
Therefore, the $s^{th}$ column of $\mathbf{G}^T$ in
(\ref{eq:BitMetricsZF}) is equal to $\lambda_s^{-1} \mathbf{u_s}$,
leading to $\mathbf{{||\mathbf{g}_s||}^2}=1/\lambda_s^{2}$. By
replacing $\mathbf{{||\mathbf{g}_s||}^2}$ and $\hat{r}_{k,s}$ in
(\ref{eq:BitMetricsZF}) with $1/\lambda_s^{2}$ and $\lambda_s^{-1}
r'_{k,s}$, respectively, the bit metrics for the ZF decoder
become
\begin{align}
\gamma_{ZF} ^i (\hat{r}_{k,s},c_{k'}) =  \min\limits_{x \in
\chi_{c_{k'}}^i} |r'_{k,s} - \lambda_s x|^2
\label{eq:BitMetricsPerfectZF}
\end{align}
which are equal to the bit metrics of BICMB in
(\ref{eq:BitMetrics}).\\

3) \emph{MMSE Receiver}

The MMSE detector $\mathbf{G}$ in (\ref{eq:LFMMSE_detector}) with
perfect CSIT becomes
\begin{align}
\mathbf{G} = \mbox{ } & [\mathbf{(HV \Phi}_N)^H\mathbf{HV
\Phi}_N+\sigma^2\mathbf{I}]^{-1}\mathbf{(HV \Phi}_N)^H \nonumber\\
= \mbox{ } & [\mathbf{\hat{\Sigma}}^H \mathbf{\hat{\Sigma}} +
\sigma^2 \mathbf{I}]^{-1} \mathbf{\hat{\Sigma}}^H \mathbf{U}^H \nonumber\\
= \mbox{ } & [\mathbf{\Sigma}_S^2 + \sigma^2 \mathbf{I}]^{-1}
\mathbf{\hat{\Sigma}}^H \mathbf{U}^H. \label{eq:MMSEDetector}
\end{align}
If we define $\mathbf{\Omega}$ as
\begin{align}
\mathbf{\Omega} = \mathbf{\Sigma}_S^2 + \sigma^2 \mathbf{I}
\label{eq:DefOmega}
\end{align}
then, $\mathbf{\Omega}$ is an $S \times S$ diagonal matrix whose
$s^{th}$ diagonal element can be expressed as $\mu_s = \lambda_s^2
+ \sigma^2$. The baseband signal after the MMSE detector given in
(\ref{eq:InOutBeamLFZF}) is
\begin{align}
\mathbf{r}_{MMSE} =\mathbf{\Omega}^{-1} \mathbf{\hat{\Sigma}}^H
\mathbf{U}^H \mathbf{y} = \mathbf{\Omega}^{-1} \mathbf{\Sigma}_S^2
\mathbf{x} + \mathbf{\Omega}^{-1} \mathbf{\hat{\Sigma}}^H
\mathbf{U}^H \mathbf{n} \label{eq:InOutBeamPerfectMMSE}
\end{align}
where $\mathbf{G}$ is replaced by the shortened form of
(\ref{eq:MMSEDetector}) and (\ref{eq:DefOmega}). Since
$\mathbf{\Omega}^{-1}$, $\mathbf{\hat{\Sigma}}$ and
$\mathbf{\Sigma}_S^2$ are all diagonal matrices, the baseband
vector signal can be separated into each subchannel signal,
resulting in the following relation with $r'_{k,s}$ of
(\ref{eq:InOutBeam_BICMB-U}) as
\begin{align}
\tilde{r}_{k,s} = \frac{\lambda_s^2}{\mu_s} x_{k,s} +
\frac{\lambda_s}{\mu_s} n_{k,s} = \frac{\lambda_s}{\mu_s}
r'_{k,s}. \label{eq:InOutBeam_Perfect-MMSE}
\end{align}

The bit metrics in (\ref{eq:BitMetricsMMSE}) require the
calculation of a matrix $\mathbf{W}$. Using an analysis similar to
the MMSE detector, $\mathbf{W}$ can be expressed as
\begin{align}
\mathbf{W} = [\mathbf{I} + \sigma^2
(\mathbf{\Sigma}_S^2)^{-1}]^{-1} \label{eq:W_SigmaS}.
\end{align}
By multiplying $(\mathbf{\Sigma}_S^2)^{-1}$ with the both sides of
(\ref{eq:DefOmega}), we get
\begin{align}
\mathbf{I} + \sigma^2 (\mathbf{\Sigma}_S^2)^{-1} = \mathbf{\Omega}
(\mathbf{\Sigma}_S^2)^{-1} \label{eq:W_Omega}.
\end{align}
Using (\ref{eq:W_SigmaS}) and (\ref{eq:W_Omega}), the $s^{th}$
diagonal element $W_{ss}$ of $\mathbf{W}$ can be easily found as
$W_{ss} = \lambda_s^2 / \mu_s$. Finally, with the help of
simplified $W_{ss}$ and the relation with $r'_{k,s}$ of
(\ref{eq:InOutBeam_Perfect-MMSE}), the bit metrics in
(\ref{eq:BitMetricsMMSE}) become
\begin{align}
\gamma_{MMSE}^i (\tilde{r}_{k,s},c_{k'}) =  \min\limits_{x \in
\chi_{c_{k'}}^i} \frac{1}{\sigma^2} |r'_{k,s}- \lambda_s x
|^2
\end{align}
which are equivalent to the bit metrics of BICMB in
(\ref{eq:BitMetrics}) because the constant $1 / \sigma^2$ can be
ignored in the metric calculation.

\end{document}